\documentclass[amsmath,amssymb,aps,prl,twocolumn,notitlepage,superscriptaddress,longbibliography]{revtex4-2}
\usepackage{graphicx}
\usepackage{amsfonts}
\usepackage{amssymb}
\usepackage{amsthm}
\usepackage{array}
\usepackage{amsmath}
\usepackage{verbatim} 
\usepackage{hyperref}   
\usepackage{xcolor}
\usepackage{bbold}
\usepackage{epstopdf}
\usepackage{mathtools}
\usepackage[normalem]{ulem}
\usepackage{physics}
\usepackage{soul}
\usepackage{bm}
\usepackage{lipsum}

\newtheorem{theorem}{Theorem}

\newtheorem{corollary}{Corollary}

\def\Tr{\mbox{Tr}}

\newcommand{\del}[1]{{\iffalse #1 \fi}}
\newcommand{\delete}[1]{{\color{red}\sout{#1}}}

\def\draft{1}
\ifnum\draft=1

\begin{document}
\title{Classical algorithms for estimating expectation values in linear-optical circuits}

\author{Youngrong Lim}
\affiliation{Department of Physics, Chungbuk National University, Cheongju, Chungbuk 28644, Korea}
\affiliation{School of Computational Sciences, Korea Institute for Advanced Study, Seoul 02455, Korea}
\author{Changhun Oh}
\email{changhun0218@gmail.com}
\affiliation{Department of Physics, Korea Advanced Institute of Science and Technology, Daejeon 34141, Korea}

\begin{abstract}
We present a classical algorithm for approximating the expectation values of observables in linear-optical circuits with arbitrary product input states, achieving additive-error accuracy.
This result indicates that current applications of photonic systems aimed at demonstrating practical quantum supremacy through expectation value estimation, such as photonic variational algorithms, may face challenges in attaining the computational advantage.
It also implies the output probabilities of boson sampling with arbitrary product input states can be efficiently approximated by our method, resulting that boson sampling becomes efficiently simulable when its output probability distribution is polynomially sparse.
We also develop an efficient classical algorithm for estimating transition amplitudes of arbitrary product states in linear-optical circuits. This provides additive-error approximation algorithms for matrix functions associated with linear-optical circuits, such as the (loop-)hafnian, which are of independent interest.
As an application, it solves the generalized molecular vibronic spectra problem~(Oh et al., 2024), previously suggested as a candidate for practical quantum advantage. Finally, we extend our framework to near-Clifford circuits, enabling classical approximation of their expectation values.
\end{abstract}
\maketitle

\textit{Introduction.---}
Linear-optical system has been considered as a promising platform for universal quantum computing~\cite{knill2001scheme} and also for demonstrating the quantum computational advantage in the near future since the seminal work on boson sampling~\cite{aaronson2011computational}. 
The boson sampling proposal~\cite{aaronson2011computational, hamilton2017gaussian, deshpande2022quantum} motivated various experiments to demonstrate the quantum advantage~\cite{zhong2020quantum,zhong2021phase,madsen2022quantum,deng2023gaussian,young2024atomic} while numerous classical algorithms have also been developed to challenge them, leading to significant advancements in understanding the computational complexity of linear-optical circuits~\cite{neville2017classical,clifford2018classical,oszmaniec2018classical,garcia2019simulating,qi2020regimes,oh2021classical,quesada2022quadratic,bulmer2022boundary,oh2022classical,oh2023classical,liu2023simulating,oh2024classical, oh2025recent, oh2025classical}.

Recently, the focus has shifted from demonstrating hardness through sampling to identifying {\it practically motivated} tasks that reflect a {\it practical} form of quantum advantage.
A common viewpoint is that near-term quantum devices aim to estimate expectation values of observables with controlled additive accuracy, a task believed to be classically hard~\cite{bravyi2021classical, daley2022practical,kim2023evidence,trivedi2024quantum,gonthier2022measurements,huang2020predicting,google2025observation}.
Within linear-optics, several proposals adopt this estimation-based framework~\cite{huh2015boson,bromley2020applications,chabaud2021quantum,yin2024experimental,hoch2025quantum}, yet no provable advantage has been established.
A notable case is molecular vibronic spectra~(MVS): while Gaussian boson sampling enables an efficient photonic procedure~\cite{huh2015boson,hamilton2017gaussian,deshpande2022quantum, shen2018quantum,paesani2019generation,wang2020efficient}, recent work shows that a classical algorithm achieves comparable performance~\cite{oh2024mvs}, because finite sampling enforces an additive-error tolerance, ruling out an advantage for MVS in this setting.

The essential issue, therefore, is precision: quantum devices estimate expectation values from finite samples, making {\it additive-error} approximations intrinsic to measurement~(see SM Sec.~S2~\cite{SM}). 
This immediately opens room for classical counterparts operating under the same additive tolerance. 
Accordingly, a rigorous assessment of practical quantum advantage requires identifying the classical limit of efficient estimability for the same observables within an additive-error tolerance. 
Although additive-error estimators for expectation values have been analyzed in other architectures~\cite{bravyi2021classical,beguvsic2023simulating,kim2023evidence,tindall2024efficient,bravyi2024classical,angrisani2024classically}, to the best of our knowledge no such provable algorithm existed for linear-optical circuits.

\begin{figure*}[t]
\includegraphics[width=\textwidth]{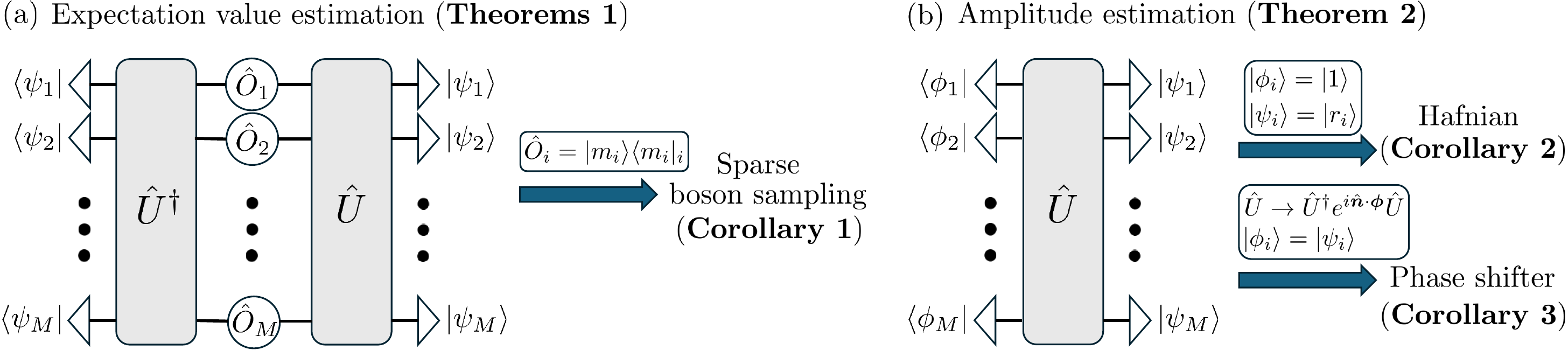}
\caption{Summary of classical algorithms introduced in the main text. (a) A schematic of the general expectation value of an $M$-mode linear-optical circuit (Theorem~\ref{th:exp}). If $\hat{O}_i$ are photon-number projectors $|m_i\rangle \langle m_i|_i$, the expectation values represent marginal output probabilities of a boson sampling circuit, resulting in an efficient simulation of a sparse boson sampling (Corollary~\ref{co:sparse}). (b) A transition amplitude of a linear-optical circuit~(Theorem~\ref{th:ggurvits}). When $|\phi_i \rangle$'s are single-photon Fock states and $|\psi_i\rangle$'s are squeezed vacuum states, the amplitude is proportional to the hafnian of a general complex symmetric matrix (Corollary~\ref{co:haf}). A transition amplitude is reduced to a phase shifter expectation value, Eq.~\eqref{eq:fourier}, with $|\phi_i\rangle=|\psi_i\rangle$ and $\hat{U} \rightarrow \hat{U}^\dagger e^{i\hat{\bm{n}}\cdot \bm{\phi}}\hat{U}$ (Corollary~\ref{co:LCBS}).}
\label{fig:scheme}
\end{figure*}

In this work, we present a classical algorithm for approximating expectation values in linear-optical circuits within additive-error.
When the observable ${\hat O=\hat O_A\otimes\mathbb{1}_B}$ has polynomially bounded Hilbert–Schmidt norm, $\|\hat O_A\|_2^2=O(\mathrm{poly}(M))$, the algorithm is efficient for arbitrary product input states.
This dequantizes a broad class of estimation tasks relevant to NISQ photonic computing that align with the framework of practical quantum advantage, including photonic variational quantum algorithms~(VQAs) where each term of the target Hamiltonian is a local observable~\cite{peruzzo2014variational,pappalardo2024photonic,hoch2024variational,facelli2024exact,baldazzi2025four,maring2024versatile}.
The same framework efficiently approximates (marginal) output probabilities of boson-sampling circuits with arbitrary product inputs.
Hence, any boson sampling is classically simulable if its output distribution is polynomially sparse~\cite{schwarz2013simulating}, generalizing earlier restricted results~\cite{roga2020classical,kolarovszki2023simulating,lim2023approximating}.
Consequently, in the polynomially sparse regime, peaks of inverse-polynomial probability can be classically identified, tempering peak-based proposals of verifiable quantum advantage based on linear optics~\cite{aaronson2022much,aaronson2024verifiable,kushilevitz1991learning,schwarz2013simulating}.


We further develop an efficient quantum-optics-inspired algorithm for transition amplitudes in linear-optical circuits with arbitrary product inputs.
While the single-photon Fock case involves permanents, classically estimable via Gurvits' algorithm~\cite{gurvits2005complexity,aaronson2011computational}, our analysis follows a distinct route with comparable guarantees and extends to quantum-inspired estimators for other matrix functions, including the (loop-)hafnian~\cite{barvinok2016combinatorics,hamilton2017gaussian,quesada2019franck}.
It also provides a complete classical solution to the generalized MVS problem~\cite{oh2024mvs} and to efficient estimation of binned output probabilities used for boson-sampling verification~\cite{drummond2022simulating,singh2023proof,seron2024efficient,anguita2025experimental}.

Finally, we apply our method to near-Clifford circuits~\cite{bravyi2016improved, bravyi2019simulation} and show that we can efficiently estimate the expectation values of observables and any (marginal) output probabilities for these circuits within an additive-error.

\textit{Approximation of expectation values.---}
Let $| \psi \rangle=\otimes_{i=1}^M |\psi_i \rangle$ be a product input state in an $M$-mode system, and $\hat{U}$ be a quantum circuit applied to the input state.
For the output state of the circuit, $\hat{U}|\psi\rangle$, we consider the expectation value of a product operator $\hat{O}=\otimes_{i=1}^M\hat{O}_i$ (see Fig.~\ref{fig:scheme}(a)):
\begin{align}\label{eq:amp}
    \langle \psi | \hat{U}^\dagger \hat{O} \hat{U} | \psi \rangle.
\end{align}
Throughout this work, input states $|\psi\rangle$ and operators $\hat{O}$ are assumed in a product form unless stated otherwise, which is practically relevant in various applications.
When an operator acts nontrivially only on a subsystem $A$ and trivially on the rest of the system $B$, we write the operator as $\hat{O}=\hat{O}_A\otimes \hat{\mathbb{1}}_B=\bigotimes_{i\in A}\hat{O}_{i}\otimes \bigotimes_{i\in B}\hat{\mathbb{1}}_{i}$ and define the reduced density matrix $\hat{\rho}_A\equiv \Tr_B[\hat{U}|\psi\rangle\langle \psi|\hat{U}^\dagger]$.

We first focus on a linear-optical circuit $\hat{U}$.
Our first main result is a classical algorithm that approximates the expectation values of operators in a linear-optical circuit:
\begin{theorem}[Expectation-value approximation in linear-optical circuit]\label{th:exp}
    Consider an $M$-mode linear-optical circuit $\hat{U}$ and an operator $\hat{O}$.
    The expectation values $\langle \psi | \hat{U}^\dagger \hat{O} \hat{U} | \psi \rangle$ can be approximated within additive-error $\epsilon$ with probability $1-\delta$ in running time $O(M^2 \lVert\hat{O}_{A}\rVert^2_2\Tr[\hat{\rho}_{A}^2] \log(1/\delta)/\epsilon^2)$, where $\|\hat{O}_A\|_2^2=\Tr[\hat{O}_A^\dagger\hat{O}_A]$.
\end{theorem}
\begin{proof}[Proof sketch]
    (See SM Sec.~S3~\cite{SM} for the full proof.)
    Assume that an operator $\hat{O}$ acts nontrivially on an $l$-mode subsystem $A$.
    Using the displacement operator expansion of the operator, $\hat{O}_A=\pi^{-l} \int d^{2l}\alpha \chi_{\hat{O}_A}(\bm{\alpha})\hat{D}^\dagger(\bm{\alpha})$~\cite{cahill1969density, ferraro2005gaussian, serafini2017quantum, SM}, where $\hat{D}(\bm{\alpha})$ is a displacement operator and $\chi_{\hat{Q}}(\bm{\alpha})\equiv \Tr[\hat{Q}\hat{D}(\bm{\alpha})]$ is the characteristic function of the operator~$\hat{Q}$, we rewrite the expectation value as
    \begin{align}
        \langle \psi|\hat{U}^\dagger \hat{O}\hat{U}|\psi\rangle 
        &=\int d^{2l}\alpha \frac{|\chi_{\hat{O}_A}(\bm{\alpha})|^2}{\pi^l\|\hat{O}_A\|_2^2}\frac{\|\hat{O}_A\|_2^2\chi_{\hat{\rho}_A}^*(\bm{\alpha})}{\chi^*_{\hat{O}_A}(\bm{\alpha})},
    \end{align}
    where $q(\bm{\alpha})\equiv |\chi_{\hat{O}_A}(\bm{\alpha})|^2/(\pi^l\|\hat{O}_A\|_2^2)$ is normalized to~1 and thus a proper probability distribution of~$\bm{\alpha}$.
    Thus, the expression implies that by sampling $\bm{\alpha}$ from $q(\bm{\alpha})$ and averaging the random variable $X(\bm{\alpha})\equiv \|\hat{O}_A\|_2^2\chi_{\hat{\rho}_A}^*(\bm{\alpha})/\chi^*_{\hat{O}_A}(\bm{\alpha})$ over the samples, we obtain an estimate of the expectation value. Thanks to the tensor-product structure of the state $\hat{\rho}_A$ and the operator $\hat{O}_A$, this process can be performed efficiently~\cite{mari2012positive,rahimi2016sufficient,footnote1}.
    
    We then show that the median-of-means estimator~\cite{jerrum1986random} yields an estimate with additive-error $\epsilon$ with probability $1-\delta$ when $O(\|\hat{O}_A\|_2^2\Tr[\hat{\rho}_{A}^2]\log\delta^{-1}/\epsilon^2)$ number of samples is given, by showing that the variance of the random variable $X(\bm{\alpha})$ is upper-bounded by $\|\hat{O}_A\|_2^2\Tr[\hat{\rho}_{A}^2]$.
    Note that the median-of-means estimator for complex random variables can be constructed by combining the real and imaginary parts estimators~\cite{SM}.
    Since it takes $O(M^2)$ time for computing  random variable $X(\bm{\alpha})$, the total running time is $O(M^2\lVert\hat{O}_{A}\rVert^2_2\Tr[\hat{\rho}_{A}^2] \log(1/\delta)/\epsilon^2)$.
\end{proof}
Therefore, we achieve an efficient additive estimation of the expectation value in polynomial time when $\lVert\hat{O}_A\rVert^2_2=O(\text{poly}(M))$.
Such a condition holds whenever $\hat{O}_A$ has $O(\text{poly}(M))$ rank with polynomially large singular values. We present two important classes of observables that satisfy the condition.

A first important class arises when $\hat{O}_A$ is \emph{local} observables, directly relevant to photonic VQAs implemented by linear-optical circuits~\cite{peruzzo2014variational,pappalardo2024photonic,hoch2024variational,facelli2024exact,baldazzi2025four,shang2024boson,maring2024versatile}.
In these schemes, the physically relevant target Hamiltonian decomposes into few-mode local terms,
\begin{align}
    \hat{H} = \sum_i \alpha_i \hat{H}_i + \sum_{i<j} \beta_{ij} \hat{H}_{ij} + \cdots,
\end{align}
and each $\hat{H}_i,\hat{H}_{ij},\dots $ acts on a small number of modes with bounded Hilbert–Schmidt norm as $O(1)$.
Our algorithm efficiently estimates the contribution of each local term, and the total energy follows by summation.
The same bounded-norm observables also appear as the primary measurement targets in other NISQ paradigms, such as analog and digital quantum simulators~\cite{ebadi2021quantum,bernien2025coarsening,zhang2022fermi,deb2022rabi,rahman2022digital,li2024ising,smith2019simulating}, where physical quantities, such as magnetizations, correlations, or energy densities, are expressed as few-body expectation values.
Hence, the condition $\|\hat{O}_A\|_2^2=O(\mathrm{poly}(M))$ precisely characterizes the physically relevant observables not only in photonic VQAs but broadly across NISQ devices, delineating a classically simulable regime for expectation-value estimation in linear optics~(See SM Sec.~S2 for more details~\cite{SM}.).

A second important class consists of observables where $\hat{O}_A$ is a rank-one projector of product form. 
This case is directly tied to projective measurements on the output of a linear-optical circuit.
A notable example is a photon-number projector in boson sampling~\cite{aaronson2011computational}.
In this case, Theorem~\ref{th:exp} implies that the (marginal) output probabilities of boson sampling with {\it arbitrary} input states can be efficiently approximated.
To the best of our knowledge, such a result has only been known for the Fock state input~\cite{aaronson2011computational, singh2023proof, oh2024mvs}.
Remarkably, this finding also enables an efficient simulation of boson sampling with \emph{arbitrary} product input if polynomially many outcomes dominate the output probabilities~\cite{kushilevitz1991learning, schwarz2013simulating,pashayan2020estimation}. (Note that sparsity of the output distribution can be determined. See SM Sec.~S9 for more details~\cite{SM}.)
\begin{corollary}[Sparse boson sampling]\label{co:sparse}
    Boson sampling with an arbitrary product input state can be efficiently simulated if its output probability distribution is polynomially sparse. 
\end{corollary}

In addition, note that quantum sampling with a peaked output probability distribution, i.e., one that has an outcome with an inverse-polynomially large probability, has drawn attention for a verifiable quantum advantage of sampling~\cite{aaronson2022much, aaronson2024verifiable} because the peak can serve as a verification. 
However, our result, together with the algorithm in Refs.~\cite{kushilevitz1991learning, schwarz2013simulating}, indicates that the peak of the output probability distribution in boson sampling with any product input can be efficiently identified using classical computers.
We emphasize that Corollary~\ref{co:sparse} is valid for a general case regardless of phase space negativity~\cite{pashayan2015estimating,lim2023approximating} or stellar rank~\cite{chabaud2023resources} of the input state, resources for quantum advantages. 

\textit{Approximation of transition amplitudes.---}
We now present another classical algorithm approximating a transition amplitude $\langle \phi|\hat{U}|\psi\rangle$ (See Fig.~\ref{fig:scheme}~(b)).
Note that when $|\phi\rangle$ and $|\psi\rangle$ are single-photon Fock states, the amplitude becomes a permanent~\cite{aaronson2011computational}, which can be efficiently approximated by Gurvits' algorithm~\cite{gurvits2005complexity}.
In contrast, our algorithm works for arbitrary product inputs.
\begin{theorem}[Transition amplitude approximation in linear-optical circuit]\label{th:ggurvits}
    The transition amplitude $\langle \phi|\hat{U}|\psi\rangle$ of a linear-optical circuit $\hat{U}$ with arbitrary states $|\psi\rangle$ and $|\phi\rangle$ can be approximated within additive-error $\epsilon$ with probability $1 - \delta$ in running time $O(M^2 \log \delta^{-1}/\epsilon^2)$.
\end{theorem}
\begin{proof}[Proof Sketch]
    We rewrite the amplitude by exploiting the completeness relation $\hat{\mathbb{1}}_{1:M}=\int d^{2M}\alpha|\bm{\alpha}\rangle \langle \bm{\alpha}|/\pi^M$:
    \begin{align}
        \langle \phi|\hat{U}|\psi\rangle
        &=\int d^{2M}\alpha \frac{\langle \phi|\hat{U}|\bm{\alpha}\rangle}{\langle\bm{\alpha}|\psi\rangle^*}\frac{|\langle \bm{\alpha}|\psi\rangle|^2}{\pi^M},
    \end{align}
    where $|\bm{\alpha}\rangle$ represents a coherent state.
    Thus, if we sample $\bm{\alpha}$ from $|\langle \bm{\alpha}|\psi\rangle|^2/\pi^M$, which is a proper probability distribution (Husimi Q-distribution of the state $|\psi \rangle$~\cite{husimi1940some}) and easy to sample from, and average $\langle \phi|\hat{U}|\bm{\alpha}\rangle/\langle \bm{\alpha}|\psi\rangle^*$ over the samples, we obtain an estimate of the amplitude.
    We prove that this procedure is efficient and gives an estimate with $\epsilon$ additive-error with probability $1-\delta$ in running time $O(M^2\log\delta^{-1}/\epsilon^2)$ in SM Sec.~S4~\cite{SM}.
\end{proof}

Note that the original Gurvits' algorithm applies only when $|\psi\rangle=|\bm{m}\rangle$, $|\phi\rangle=|\bm{n}\rangle$ with $\bm{n},\bm{m}\in \{0,1\}^M$, and its variants~\cite{aaronson2012generalizing, oh2024mvs} apply when $\bm{m}\in \{0,1\}^M$ or $\bm{n}=\bm{m}$, corresponding to approximating the permanent of matrices with repeated rows and columns in a certain pattern.
Notably, our algorithm not only extends to computing the permanent of matrices with repeated rows and columns in a general pattern, but also introduces a new approach to approximating the permanent using a quantum-optics-inspired method.

Our algorithm can also approximate the hafnian~\cite{barvinok2016combinatorics}, a generalization of the permanent, by using the fact that the hafnian of any $M \times M$ complex symmetric matrix $R$ ($M$ is even) with the largest singular value $\lambda_\text{max}<1$ can be expressed as a transition amplitude of an $M$-mode Gaussian boson sampling circuit~\cite{hamilton2017gaussian,quesada2019franck}:
\begin{align}\label{eq:haf}
    \text{haf}(R)=\frac{1}{Z^{1/2}}\langle \bm{1} | \hat{U} | \bm{r}\rangle,
\end{align}
where $|\bm{r}\rangle =\otimes_{i=1}^M |r_i\rangle$ denotes a product of squeezed vacuum states $| r_i \rangle$ of squeezing parameter $r_i$, $|\bm{1}\rangle=\otimes_{i=1}^M |1 \rangle$ is $M$ single-photon states, and $Z\equiv \prod_{i=1}^M\cosh{r_i}=\prod_{i=1}^M(1-\lambda_i^{ 2})^{-\frac{1}{2}}$ with $R$'s singular values $\lambda_i$.
Until now, efficient additive approximation algorithms have been known only for restricted cases~\cite{lim2023approximating,oh2024mvs}. Notably, our generalized Gurvits' algorithm leads to an additive approximation algorithm for the hafnian without any constraints, which can be seen by substituting $|\psi\rangle=|\bm{r}\rangle$, $|\phi\rangle=|\bm{1}\rangle$ in Theorem~\ref{th:ggurvits}:
\begin{corollary}[Hafnian approximation]\label{co:haf}
    The hafnian of an $M \times M$ complex symmetric matrix $R$ can be approximated within additive-error $\epsilon \norm{R}^{M/2}$ with probability $1-\delta$ in $O(M^3\log \delta^{-1}/\epsilon^2)$ time.
\end{corollary}
A detailed proof is provided in SM Sec.~S5~\cite{SM}. By a similar method, the loop hafnian of an arbitrary complex symmetric matrix~\cite{quesada2019franck} can also be efficiently approximated.

It is worth emphasizing that Theorem~\ref{th:ggurvits} can also be used to approximate the overlap of two quantum states generated by linear-optical circuits.
Such a quantity has attracted attention for machine learning applications of boson sampling, such as kernel method \cite{chabaud2021quantum, yin2024experimental, hoch2025quantum}.
More specifically, boson samplers are used to estimate the overlap of two quantum states corresponding to two different data vectors $\bm{x},\bm{y}$ as $|\langle \psi|\hat{U}^\dagger(\bm{x})\hat{U}(\bm{y})|\psi\rangle|^2$, where $\hat{U}(\bm{x})$ corresponds to the linear-optical circuits whose configuration depend on the data vector $\bm{x}$.
Thus, our classical method can replace boson samplers for this purpose.

\textit{Phase shifter.---}
We now consider a problem of approximating the expectation value of a phase shifter $e^{i\hat{\bm{n}}\cdot \bm{\phi}}$ of a linear-optical circuit~\cite{ivanov2020complexity}
\begin{equation}\label{eq:fourier}
    \langle \psi|\hat{U}^{\dagger}e^{i\hat{\bm{n}}\cdot \bm{\phi}}\hat{U}|\psi\rangle,
\end{equation}
where $\hat{\bm{n}}\equiv(\hat{n}_1,\dots,\hat{n}_M)$ and $\bm{\phi}\in\mathbb{R}^M$ represents the phase vector.
The expectation value of a phase shifter has various applications such as the MVS problem~\cite{huh2015boson, oh2024mvs} and approximating binned probability, proposed for verifying boson sampling experiments~\cite{drummond2022simulating,singh2023proof,seron2024efficient, anguita2025experimental}.

In particular, the MVS problem~\cite{huh2015boson} has been considered as a potential application of boson sampling, which is essentially equivalent to computing grouped output probabilities $p(\bm{m})$ of boson sampling~(See SM~Sec.~S4~\cite{SM} for more details):
\begin{align}\label{eq:group}
    G(\Omega)\equiv \sum_{\bm{m}=\bm{0}}^{\infty}p(\bm{m})\delta(\Omega-\bm{\omega}\cdot \bm{m}),
\end{align}
where $\bm{\omega} \in \mathbb{Z}^M_{\geq0}$ and $\Omega=\{0,\dots,\Omega_\text{max}\}$.
In Ref.~\cite{oh2024mvs}, however, the authors observe that the Fourier components of Eq.~\eqref{eq:group} are expressed as the expectation value of a phase shifter as Eq.~\eqref{eq:fourier}. Therefore, if Eq.~\eqref{eq:fourier} is additively approximated, the grouped probabilities Eq.~\eqref{eq:group} are reproduced within an additive-error through the inverse Fourier transform~\cite{SM}.


While Ref.~\cite{oh2024mvs} proves that this is the case when $|\psi\rangle$ is a product Gaussian state or a Fock state, leading to dequantization, for more general input states, like a product squeezed Fock state, a quantum advantage may still exist since no classical algorithm is known yet.
This possibility is particularly important because such an extension has recently been experimentally implemented~\cite{wang2020efficient}.
Notably, by setting $|\phi \rangle=|\psi\rangle$ and $\hat{U}\to \hat{U}^{\dagger}e^{i\hat{\bm{n}}\cdot \bm{\phi}}\hat{U}$ in Theorem~\ref{th:ggurvits}, our classical algorithm can approximate Eq.~\eqref{eq:fourier} for any product input state:
\begin{corollary}[Phase-shifter expectation value]\label{co:LCBS}
    Phase-shifter expectation value of a linear-optical circuit can be efficiently approximated for an arbitrary product input state.
\end{corollary}
Therefore, our algorithm enables efficient approximation of the grouped probabilities in Eq.~\eqref{eq:group} for arbitrary product input states and thus solves a generalized MVS problem for arbitrary product input states, an open problem posed in Ref.~\cite{oh2024mvs}.

\textit{Extending to qubit circuits.---}
Finally, we show that our classical algorithm in Theorem~\ref{th:exp} can be extended to qubit system, especially near-Clifford circuits $\hat{U}_\text{NC}$~\cite{bravyi2016improved, bravyi2019simulation}, i.e., Clifford circuit $\hat{U}_\text{C}$+logarithmic number of $T$ gates, with arbitrary product input states. It is known that both exact and approximate classical simulations of such circuits are computationally hard, assuming the non-collapse of the polynomial hierarchy~\cite{jozsa2013classical}.
Consequently, to investigate routes toward practical quantum advantage, our method provides a direct additive-error estimation scheme for the expectation value of an observable, $\langle \psi | \hat{U}^\dagger_{\text{NC}}\hat{O}\hat{U}_{\text{NC}} |\psi\rangle$ that does not rely on sampling from the output distribution, thereby bypassing the main bottleneck that underlies the known sampling hardness results.

Crucial properties used in Theorem~\ref{th:exp} are that a displacement operator is transformed into another displacement operator through linear-optical circuits and that any operator can be expanded in terms of displacement operators.
Similar properties hold for Clifford circuits where any operators can be decomposed using Pauli operators and a Pauli operator transforms another Pauli operator under Clifford circuits.
Thus, we obtain a similar result:
\begin{theorem}[Expectation value approximation in near-Clifford circuit]\label{th:nc}
    Consider an $n$-qubit near-Clifford circuit $\hat{U}_{\text{NC}}$ with depth $O(\text{poly}(n))$ and an operators $\hat{O}$.
    The expectation values $\langle \psi | \hat{U}^\dagger_{\text{NC}}\hat{O}\hat{U}_{\text{NC}} |\psi\rangle$ can be approximated within additive-error $\epsilon$ with probability $1-\delta$ in running time $O(\text{poly}(n)\lVert\hat{O}_{A}\rVert^2_2\Tr[\hat{\rho}_{A}^2] \log(1/\delta)/\epsilon^2)$.
\end{theorem}
\begin{proof}[Proof Sketch]
    The proof is similar to that of Theorem~\ref{th:exp} and provided in SM~Sec.~S6~\cite{SM}.
    Let us first focus on Clifford circuits.
    Using the Pauli operator expansion, $\hat{O}_A=2^{-l}\sum_{\bm{a}\in \{0,1,2,3\}^{l}}\chi_{\hat{O}_A}(\bm{a})\hat{P}_{\bm{a}}$, where $\chi_{\hat{Q}}(\bm{a})\equiv \Tr[\hat{Q}\hat{P}_{\bm{a}}]$ and $\hat{P}_{\bm{a}}\in \{I,X,Y,Z\}^l$ is a Pauli operator, the expectation value of $\hat{O}=\hat{O}_{A}\otimes \hat{\mathbb{1}}_{B}$ can be written as 
    \begin{align}
        \langle \psi | \hat{U}^\dagger_\text{C} \hat{O}\hat{U}_\text{C} | \psi 
        =\sum_{\bm{a} \in \{ 0,1,2,3\}^l}\frac{|\chi_{\hat{O}_A}(\bm{a})|^2}{2^l\|\hat{O}_A\|_2^2}\frac{\|\hat{O}_A\|_2^2\chi_{\hat{\rho}_A}(\bm{a})}{\chi^*_{\hat{O}_A}(\bm{a})}.
    \end{align} 
    Thus, by sampling $\bm{a}$ from $q(\bm{a})\equiv |\chi_{\hat{O}_A}(\bm{a})|^2/(2^l\|\hat{O}_A\|_2^2)$ and averaging the random variable $X(\bm{a})\equiv \|\hat{O}_A\|_2^2\chi_{\hat{\rho}_A}(\bm{a})/\chi^*_{\hat{O}_A}(\bm{a})$ over the samples, we obtain an estimate of the expectation value.
    By showing that the variance of the random variable $X(\bm{a})$ is bounded by $\lVert\hat{O}_{A}\rVert^2_2\Tr[\hat{\rho}_{A}^2]$, the median-of-means estimator gives the desired error presented in the theorem.

    Furthermore, our method is still valid for the additional logarithmic number of $T$-gates because a single Pauli operator is transformed into the sum of at most two Pauli operators after passing through each $T$-gate, that is, $\hat{T}^\dagger \hat{P}_0\hat{T}=\hat{P}_0,~\hat{T}^\dagger \hat{P}_1\hat{T}=(\hat{P}_1+\hat{P}_2)/\sqrt{2},~\hat{T}^\dagger \hat{P}_2\hat{T}=(\hat{P}_2-\hat{P}_1)/\sqrt{2},~\hat{T}^\dagger \hat{P}_3\hat{T}=\hat{P}_3$. If there are $t$ of $T$-gates in a near-Clifford circuit $\hat{U}_\text{NC}$, we have at most $n2^t$ Pauli operators to track the transformations.
Therefore, if $t=O(\log n)$, we can handle all the relevant Pauli operators in $O(\text{poly}(n))$ time. A detailed proof is given in SM Sec.~S7~\cite{SM}.
\end{proof}

Theorem~\ref{th:nc} implies that the expectation value of observables $\hat{O}$ with $\lVert \hat{O} \rVert_2^2=O(\text{poly}(n))$ of a near-Clifford circuit with an arbitrary product input state can be efficiently approximated. 
Especially for rank-1 projective measurements, this result implies that the magic resources of both the input state and measurement do not increase the complexity of additive-error-approximation.
Note that our result far extends the previous results: the expectation values of a near-Clifford circuit with stabilizer input or for Pauli observables~\cite{bravyi2016improved, bravyi2019simulation}, and output probabilities in the computational basis~\cite{aaronson2004improved,bravyi2019simulation,pashayan2020estimation}. (See SM Sec. S10 for more details~\cite{SM}.)

\textit{Discussion.---}
We presented a classical algorithm for estimating expectation values of observables in linear-optical circuits.
Using this, we showed that one can efficiently approximate (marginal) output probabilities of boson sampling with arbitrary product input states, which implies that the simulation of sparse boson sampling with arbitrary product input states is easy.
We also provided an efficient classical algorithm for the general transition amplitude in a linear-optical circuit, which allows us to approximate the hafnian of any complex symmetric matrix. 
Furthermore, using this classical algorithm, we solved an open problem in Ref.~\cite{oh2024mvs}, namely, a generalized MVS problem for arbitrary product input states. 

A natural open question is to extend our method to more complicated optical circuits with nonlinear effects such as nonlinear gates or postselection. Since such cases ultimately enable universal quantum computation~\cite{lloyd1999quantum, knill2001scheme}, we do not expect our method to work for the most general cases. 
Nonetheless, it suggests a phase transition from easiness to hardness when introducing nonlinear effects; thus, an interesting open question is to determine the exact transition point.

Also, we have focused on computing grouped probabilities of boson sampling with a certain pattern in the MVS problem, which is defined by the linear relation $\Omega-\bm{\omega}\cdot\bm{m}$ in Eq.~\eqref{eq:group}.
However, our method does not straightforwardly generalize to a more complicated pattern of groups, which still opens a possibility of quantum advantage.
Thus, the complexity of approximating general grouped boson sampling output probabilities is another open question.

\begin{acknowledgments}
\textit{Acknowledgments.---}
This work was supported by the National Research Foundation of Korea Grants (No. RS-2024-00431768 and No. RS-2025-00515456) funded by the Korean government (Ministry of Science and ICT~(MSIT)), the Institute of Information \& Communications Technology Planning \& Evaluation (IITP) Grants funded by the Korea government (MSIT) (No. IITP-2025-RS-2025-02283189 and IITP-2025-RS-2025-02263264), and by Creation of the quantum information science R\&D ecosystem (based on human resources) through the National Research Foundation of Korea (NRF) funded by the Korean government (Ministry of Science and ICT (MSIT)) (RS-2023-00256050).
\end{acknowledgments}

\bibliography{reference}

\end{document}


\title{Supplemental Material for ``Classical algorithms for estimating expectation values in linear-optical circuits''}

\author{Youngrong Lim}
\affiliation{School of Computational Sciences, Korea Institute for Advanced Study, Seoul 02455, Korea}
\author{Changhun Oh}
\email{changhun0218@gmail.com}
\affiliation{Department of Physics, Korea Advanced Institute of Science and Technology, Daejeon 34141, Korea}

\pacs{}
\maketitle

\section{Review of quantum optics and boson sampling}
\subsection{Basics of quantum optics}
In this work, our main results require some background about theoretical tools for quantum optics.
We present the main ingredients to understand the formalism, which can be easily found in quantum optics literature~(e.g., Refs.~\cite{cahill1969density, ferraro2005gaussian, serafini2017quantum}).
First of all, we consider $M$-mode bosonic systems, which are defined by their annihilation and creation operators $\{\hat{a}_i\}_{i=1}^M$ and $\{\hat{a}_i^\dagger\}_{i=1}^M$, respectively.
The annihilation and creation operators follow the canonical commutation relations $[\hat{a}_i,\hat{a}_j^\dagger]=\delta_{ij}$ and $[\hat{a}_i,\hat{a}_j]=0$ for all $1\leq i,j\leq M$.
Now, the Fock states can be formally written by the creation operators as
\begin{align}
    |\bm{n}\rangle\equiv |n_1,\dots,n_M\rangle=\prod_{i=1}^M\left(\frac{\hat{a}_i^{\dagger n_i}}{\sqrt{n_i!}}\right)|0,\dots,0\rangle,
\end{align}
where $|0\rangle$ is the vacuum state.
Our main focus in this work is linear-optical circuits, which are unitary circuits that can be defined by how they transform the bosonic creation operators:
\begin{align}\label{eq:LON}
    \hat{U}\hat{a}_i^\dagger\hat{U}^\dagger
    =\sum_{j=1}^M U_{ji}\hat{a}_j^\dagger,
\end{align}
where $U$ is an $M\times M$ unitary matrix corresponding to the linear-optical circuit $\hat{U}$.

Let us now define displacement operators, which are again unitary operations.
Formally, a single-mode displacement operator is defined as $\hat{D}(\alpha)\equiv e^{\alpha \hat{a}^\dagger-\alpha^*\hat{a}}$ with $\alpha\in\mathbb{C}$.
A multi-mode displacement operator can be defined as a tensor product of single-mode displacement operators, $\hat{D}(\bm{\alpha})\equiv \otimes_{i=1}^M\hat{D}(\alpha_i)$, where $\bm{\alpha}\equiv (\alpha_1,\dots,\alpha_M)\in\mathbb{C}^M$.
An important property of displacement operators we exploit in the proof of Theorem~1 is that they form a complete basis.
More specifically, using displacement operators, we can expand any operator $\hat{O}$ as
\begin{align}\label{eq:dis_expand}
    \hat{O}=\frac{1}{\pi^M}\int d^{2M}\alpha \chi_{\hat{O}}(\bm{\alpha}) \hat{D}^\dagger(\bm{\alpha}),
\end{align}
where we defined the characteristic function of the operator $\hat{O}$:
\begin{align}
    \chi_{\hat{O}}(\bm{\alpha})\equiv \Tr[\hat{D}(\bm{\alpha})\hat{O}].
\end{align}
Hence, it can be used to express $\Tr[\hat{A}\hat{B}]$ of two Hermitian operators $\hat{A}$ and $\hat{B}$ as
\begin{align}
    \Tr[\hat{A}\hat{B}]
    =\frac{1}{\pi^M}\int d^{2M}\alpha \Tr[\hat{D}(\bm{\alpha})\hat{A}]\Tr[\hat{D}^\dagger(\bm{\alpha})\hat{B}]
    =\frac{1}{\pi^M}\int d^{2M}\alpha \chi_{\hat{A}}(\bm{\alpha})\chi^*_{\hat{B}}(\bm{\alpha}).
\end{align}
We also need to know how a displacement operator $\hat{D}(\bm{\alpha})$ transforms under a linear-optical circuit $\hat{U}$ for the proof of Theorem~1.
By the unitarity of the matrix $U$, Eq.~\eqref{eq:LON} is equivalent to
\begin{align}
    \hat{U}^\dagger \hat{a}_i^\dagger\hat{U}=\sum_{j=1}^M U_{ij}^*\hat{a}_j^\dagger,
\end{align}
and thus displacement operators transform as
\begin{align}\label{eq:dis}
    \hat{U}^\dagger \hat{D}(\bm{\alpha})\hat{U}
    =\hat{U}^\dagger e^{\sum_{i=1}^M (\alpha_i\hat{a}^\dagger_i-\alpha^*_i\hat{a}_i)} \hat{U}
    =e^{\sum_{i,j=1}^M (\alpha_i U_{ij}^*\hat{a}^\dagger_j-\alpha^*_iU_{ij}\hat{a}_j)}
    =e^{\sum_{i=1}^M ((U^\dagger\bm{\alpha})_{j}\hat{a}^\dagger_j-(U^\dagger\bm{\alpha})_{j}^*\hat{a}_j)}
    =\hat{D}(U^\dagger\bm{\alpha}).
\end{align}

Meanwhile, for Theorem 2, we use coherent states instead of displacement operators.
A single-mode coherent state $|\alpha\rangle$ can be defined as the eigenstate of the annihilation operator:
\begin{align}
    \hat{a}|\alpha\rangle=\alpha|\alpha\rangle,
\end{align}
with eigenvalue $\alpha\in \mathbb{C}$.
It can also be defined by the displacement operator as $|\alpha\rangle\equiv \hat{D}(\alpha)|0\rangle$.
The multimode coherent state is simply a product of single-mode coherent states 
\begin{align}
    |\bm{\alpha}\rangle\equiv |\alpha_1,\dots,\alpha_M\rangle=\hat{D}(\bm{\alpha})|0,\dots,0\rangle=e^{\sum_{i=1}^M (\alpha_i\hat{a}^\dagger_i-\alpha^*_i\hat{a}_i)}|0,\dots,0\rangle.
\end{align}
A crucial property that we use for Theorem~2 is that coherent states transform into coherent states under linear-optical circuits, which can be seen as 
\begin{align}\label{eq:coherent}
    \hat{U}|\bm{\alpha}\rangle
    =\hat{U}e^{\sum_{i=1}^M (\alpha_i\hat{a}^\dagger_i-\alpha^*_i\hat{a}_i)}|0\rangle 
    =\hat{U}e^{\sum_{i=1}^M (\alpha_i\hat{a}^\dagger_i-\alpha^*_i\hat{a}_i)}\hat{U}^\dagger \hat{U}|0\rangle
    =e^{\sum_{i,j=1}^M (\alpha_i U_{ji}\hat{a}^\dagger_j-\alpha^*_iU_{ji}^*\hat{a}_j)} |0\rangle
    =|U\bm{\alpha}\rangle,
\end{align}
where we used the fact that $\hat{U}|0,\dots,0\rangle=|0,\dots,0\rangle$.
Another important property of coherent states we exploit in the proof of Theorem~2 is that they form an overcomplete basis so that
\begin{align}
    \hat{\mathbb{1}}=\frac{1}{\pi}\int d^2\alpha |\alpha\rangle\langle \alpha|,
\end{align}
and for $M$ modes,
\begin{align}\label{eq:coherent_complete}
    \hat{\mathbb{1}}\otimes \dots \otimes \hat{\mathbb{1}}=\frac{1}{\pi^M}\int d^{2M}\alpha |\bm{\alpha}\rangle\langle \bm{\alpha}|.
\end{align}

\subsection{Boson sampling}
Boson sampling is an important application of linear-optical circuits, which is believed to be hard to classically simulate~\cite{aaronson2011computational}.
The hardness result has attracted a lot of attention because it can be used to demonstrate a quantum computational advantage in the near future.

The Fock boson sampling input state is given by the Fock state; more specifically, the input state is assumed to be the $N$ single photon state:
\begin{align}
    |\psi\rangle=|1,\dots,1,0,\dots,0\rangle.
\end{align}
The boson sampling circuit is given by a linear-optical circuit $\hat{U}$, and finally, the measurement is on the photon number basis.
Therefore, the output probability of obtaining photon number $\bm{m}$ can be expressed as
\begin{align}
    p(\bm{m})=|\langle \bm{m}|\hat{U}|\psi\rangle|^2.
\end{align}
Here, a crucial fact about boson sampling is that the output probability amplitude is expressed as the permanent of a matrix, which is known to be \#P-hard to compute~\cite{aaronson2011computational}.
While computing the output probability amplitude exactly is hard, it turns out that approximating them within an additive-error $\epsilon$ takes only $O(M^2/\epsilon^2)$ computational time by using Gurvits' algorithm if $\bm{m}$ is composed of $0$ or $1$~\cite{gurvits2005complexity, aaronson2011computational, aaronson2012generalizing}.
Gurvits' algorithm has been generalized to the case when $\bm{m}$ is a general Fock state or when $|\psi\rangle=|\bm{m}\rangle$; however, it has not been known if Gurvits' algorithm can be further generalized to arbitrary input state or even arbitrary measurement.
This is especially because Gurvits' algorithm exploits the structure of the permanent.
In Theorem 2 of the main text, we show that our classical algorithm can approximate $\langle \phi|\hat{U}|\psi\rangle$ for arbitrary product states $|\psi\rangle$ and $|\phi\rangle$.

\section{additive-error approximation of a product observable}
In this section, we provide the context for considering an \emph{additive-error} approximation of a \emph{product} observable.

Most importantly, additive-error is the \emph{intrinsic precision} furnished by quantum devices themselves. 
To be more specific, any estimate produced by a quantum processor is obtained by repeated measurements (i.e., sampling) and thus obeys the concentration laws.
Concretely, if one estimates an expectation value $\mathbb{E}[X]$ of a bounded random variable ($X\in[a,b]$), Hoeffding/Chernoff bounds imply an additive concentration
\begin{align}
    \Pr\left[|\mu-\mathbb{E}[X]|>\epsilon\right] \le 2\exp\left(-\frac{2N\epsilon^2}{(b-a)^2}\right),
\end{align}
so the natural sample complexity scales as ($N=O\left({(b-a)^2}/{\epsilon^2}\right)$).
Thus, in quantum devices, if each shot returns an outcome bounded by the operator norm $\|\hat O\|_\text{op}$, then the empirical average $\mu$ satisfies $|\mu-\langle \hat O\rangle|=O\left(\|\hat O\|_\text{op}/\sqrt{N}\right)$ (with high probability). Therefore, even \emph{ideal} quantum implementations provide additive (not multiplicative) guarantees limited by shot noise and finite sampling. 
As such, establishing additive-error approximability directly addresses what matters in practice across a broad range of applications while remaining exactly commensurate with the precision scale of quantum experimentation.

On the other hand, multiplicative estimation of an observable is often hard even for quantum devices.
A particularly relevant example is the output probability of a quantum circuit.
In general, this quantity is known to be \#P-hard in the worst case and thus we do not expect a quantum computer to efficiently estimate this quantity with multiplicative error~\cite{bouland2019complexity}.

For this reason, practical workflows that compute expectation values are naturally specified in absolute tolerances and evaluated under additive-error guarantees: photonic variational quantum algorithms operate on sums of few-mode local terms, with optimizers using per-term absolute targets (e.g., $10^{-2}-10^{-3}$); this is precisely the regime realized in dual-rail/interferometric photonic encodings and their recent gradient/implementation advances (Refs.~\cite{peruzzo2014variational,pappalardo2024photonic,hoch2024variational,facelli2024exact,baldazzi2025four,shang2024boson,maring2024versatile}). 
Analog quantum simulators read out few-site observables, single-site magnetizations, densities, two-point correlators, and make phase/threshold decisions using absolute criteria (e.g., Refs.~\cite{ebadi2021quantum,bernien2025coarsening,zhang2022fermi,deb2022rabi}). 
Digital quantum simulators similarly assess Trotterized circuits via local observables (edge magnetizations, short-range correlators, local energy densities), again with additive accuracy per observable (e.g., Refs.~\cite{rahman2022digital,li2024ising,smith2019simulating}). Device verification/validation tasks (e.g., vibronic spectra and binned/marginalized observables) are explicitly expectation-value problems, Fourier components or bin indicators, whose natural metric is additive-error per bin/component (e.g., Refs.~\cite{huh2015boson, oh2024mvs, drummond2022simulating,singh2023proof,seron2024efficient, anguita2025experimental}). 
Finally, quantum-enhanced photonic ML likewise estimates overlaps/kernels as expectation values and tolerates absolute errors per kernel entry (e.g., Refs.~\cite{chabaud2021quantum, yin2024experimental, hoch2025quantum}).

\renewcommand{\thetheorem}{1}
\section{Proof of Theorem 1}
In this section, we provide the proof of Theorem 1 in the main text, which is restated here:
\begin{theorem}[Expectation value approximation in linear-optical circuit]\label{th:exp}
    Consider an $M$-mode linear-optical circuit $\hat{U}$ and an operator $\hat{O}$.
    The expectation value $\langle \psi | \hat{U}^\dagger \hat{O} \hat{U} | \psi \rangle$ can be approximated within additive-error $\epsilon$ with probability $1-\delta$ in running time $O(M^2 \lVert\hat{O}_{A}\rVert^2_2 \Tr[\hat{\rho}_{A}]^2 \log(1/\delta)/\epsilon^2)$.
\end{theorem}
\begin{proof}
    First of all, we use the expansion of the operator $\hat{O}_A$ in terms of displacement operators in $l$-mode subsystem as Eq.~\eqref{eq:dis_expand}, i.e., $\hat{O}_A=\pi^{-l}\int d^{2l}\alpha \chi_{\hat{O}_A}(\bm{\alpha})\hat{D}^\dagger(\bm{\alpha})$.
    Then the expectation value can be expressed as
    \begin{align}
        \langle \psi|\hat{U}^\dagger (\hat{O}_{A}\otimes \hat{\mathbb{1}}_{B})\hat{U}|\psi\rangle 
        &=\frac{1}{\pi^l}\int d^{2l}\alpha \chi_{\hat{O}_{A}}(\bm{\alpha})\chi_{\hat{\rho}_{A}}^*(\bm{\alpha}) \\ 
        &=\frac{1}{\pi^l}\int d^{2l}\alpha |\chi_{\hat{O}_{A}}(\bm{\alpha})|^2\frac{\chi_{\hat{\rho}_{A}}^*(\bm{\alpha})}{\chi^*_{\hat{O}_{A}}(\bm{\alpha})} \\ 
        &=\int d^{2l}\alpha \frac{|\chi_{\hat{O}_{A}}(\bm{\alpha})|^2}{\pi^l\|\hat{O}_{A}\|_2^2}\frac{\|\hat{O}_{A}\|_2^2\chi_{\hat{\rho}_{A}}^*(\bm{\alpha})}{\chi^*_{\hat{O}_{A}}(\bm{\alpha})} \\ 
        &\equiv \int d^{2l}\alpha q(\bm{\alpha})X(\bm{\alpha}),
    \end{align}
    where we defined $\hat{\rho}_A\equiv \Tr_B[\hat{U}|\psi\rangle\langle \psi|\hat{U}^\dagger]$,
    \begin{align}
        q(\bm{\alpha})\equiv \frac{|\chi_{\hat{O}_{A}}(\bm{\alpha})|^2}{\pi^l\|\hat{O}_{A}\|_2^2}~~~\text{and}~~~
        X(\bm{\alpha})\equiv \frac{\|\hat{O}_{A}\|_2^2\chi_{\hat{\rho}_{A}}^*(\bm{\alpha})}{\chi^*_{\hat{O}_{A}}(\bm{\alpha})}.
    \end{align}
    Since
    \begin{align}
        \frac{1}{\pi^l}\int d^{2l}\alpha |\chi_{\hat{O}_{A}}(\bm{\alpha})|^2
        =\Tr[\hat{O}_A^\dagger\hat{O}_A]
        =\|\hat{O}_{A}\|_2^2,
    \end{align}
    $q(\bm{\alpha})$ is a proper probability distribution.
    We can easily sample $\bm{\alpha}\in\mathbb{C}^l$ from this distribution modewise because it can be decomposed as
    \begin{align}
        q(\bm{\alpha})
        =\frac{|\chi_{\hat{O}_{A}}(\bm{\alpha})|^2}{\pi^l\|\hat{O}_{A}\|_2^2}
        =\prod_{i\in A}\frac{|\chi_{\hat{O}_{i}}(\alpha_i)|^2}{\pi\|\hat{O}_{i}\|_2^2}.
    \end{align}
    Therefore, the sampling cost per sample is $O(M)$.
    Also, we can compute the random variable as
    \begin{align}
        X(\bm{\alpha})
        =\frac{\|\hat{O}_{A}\|_2^2\chi_{\hat{\rho}_{A}}^*(\bm{\alpha})}{\chi^*_{\hat{O}_{A}}(\bm{\alpha})}
        =\prod_{i\in A}\frac{\|\hat{O}_{i}\|_2^2}{\chi^*_{\hat{O}_{i}}(\alpha_i)}\prod_{i=1}^M \chi^*_{\psi_i}((U^\dagger\bm{\alpha}')_i),
    \end{align}
    where we defined $\bm{\alpha}'\equiv (\bm{\alpha},\bm{0})\in \mathbb{C}^M$ and we used
    \begin{align}
        \chi_{\hat{\rho}_{A}}(\bm{\alpha})
        &=\langle \psi|\hat{U}^\dagger(\hat{D}(\bm{\alpha})\otimes \hat{\mathbb{1}}_{B})\hat{U}|\psi\rangle \\
        &=\langle \psi|\hat{U}^\dagger(\hat{D}(\bm{\alpha})\otimes \hat{D}(\bm{0}_{B})\hat{U}|\psi\rangle \\
        &=\langle \psi|\hat{D}(U^\dagger\bm{\alpha}')|\psi\rangle \\
        &=\prod_{i=1}^M \langle \psi_i|\hat{D}((U^\dagger\bm{\alpha}')_i)|\psi_i\rangle \\ 
        &=\prod_{i=1}^M \chi_{\psi_i}((U^\dagger\bm{\alpha}')_i),
    \end{align}
    where we used Eq.~\eqref{eq:dis} for the third equality.
    Thus, the computational cost of computing a random variable $X(\bm{\alpha})$ is $O(M^2)$ from the matrix-vector multiplication.

    Instead of averaging random variable $X(\bm{\alpha})$, we employ the median-of-means estimator, which enables us to amplify the success probability and thus reuse the samples for many observables.
    For $N\geq 544\log(\frac{1}{\delta})\frac{\text{Var}(X)}{\epsilon^2}$ number of samples, the median-of-means estimator $\mu$ satisfies~\cite{jerrum1986random} (see Sec.~\ref{sec:MoM})
    \begin{align}
        \Pr[|\mu-\langle \psi|\hat{U}^\dagger (\hat{O}_{A}\otimes \hat{\mathbb{1}}_{B})\hat{U}|\psi\rangle |\geq \epsilon]\leq \delta.
    \end{align}
    Here, the variance of $X(\bm{\alpha})$, $\text{Var}(X)$, is bounded above as
    \begin{align}
        \text{Var}(X)&\leq \mathbb{E}_{\bm{\alpha}}\left[ \left|\frac{\|\hat{O}_{A}\|_2^2\chi^*_{\hat{\rho}_{A}}(\bm{\alpha})}{\chi^*_{\hat{O}_{A}}(\bm{\alpha})}\right|^2\right]=\frac{1}{\pi^{l}}\int d^{2l}\alpha \|\hat{O}_{A}\|_2^2|\chi_{\hat{\rho}_{A}}(\bm{\alpha})|^2
        =\|\hat{O}_{A}\|_2^2\|\hat{\rho}_{A}\|_2^2.
    \end{align}
    Therefore, we can achieve additive-error $\epsilon$ with probability $1-\delta$ using $N=O(\lVert\hat{O}_{A}\rVert^2_2\|\hat{\rho}_{A}\|_2^2 \log\delta^{-1}/\epsilon^2)$ sample. Since it takes $O(M^2)$ time for computing  random variable $X(\bm{\alpha})$, the total running time is $O(M^2 \lVert\hat{O}_{A}\rVert^2_2 \|\hat{\rho}_{A}\|_2^2\log(1/\delta)/\epsilon^2)$.

\end{proof}

\renewcommand{\thetheorem}{2}
\section{Proof of Theorem 2}
In this section, we provide the proof of Theorem 2 in the main text, which is restated here:
\begin{theorem}[Generalized Gurvits' algorithm]\label{th:ggurvits}
    The probability amplitude $\langle \phi|\hat{U}|\psi\rangle$ of a linear-optical circuit $\hat{U}$ with arbitrary states $|\psi\rangle$ and $|\phi\rangle$ can be approximated within additive-error $\epsilon$ with probability $1 - \delta$ in running time $O(M^2 \log \delta^{-1}/\epsilon^2)$.
\end{theorem}
\begin{proof}
    We can rewrite the amplitude by exploiting the completeness of coherent states, Eq.~\eqref{eq:coherent_complete}, i.e., $\hat{\mathbb{1}}_{1:M}=\int d^{2M}\alpha|\bm{\alpha}\rangle \langle \bm{\alpha}|/\pi^M$:
    \begin{align}
        \langle \phi|\hat{U}|\psi\rangle
        =\frac{1}{\pi^M}\int d^{2M}\alpha \langle \phi|\hat{U}|\bm{\alpha}\rangle\langle \bm{\alpha}|\psi\rangle 
        =\int d^{2M}\alpha \frac{\langle\phi |\hat{U}|\bm{\alpha}\rangle}{\langle\bm{\alpha}|\psi\rangle^*}\frac{|\langle \bm{\alpha}|\psi\rangle|^2}{\pi^M}
        \equiv \int d^{2M}\alpha X(\bm{\alpha})q(\bm{\alpha}),
    \end{align}
    where $|\bm{\alpha}\rangle$ represents a coherent state and we defined
    \begin{align}
        q(\bm{\alpha})\equiv \frac{|\langle \bm{\alpha}|\psi\rangle|^2}{\pi^M}~~~\text{and}~~~X(\bm{\alpha})\equiv \frac{\langle\phi |\hat{U}|\bm{\alpha}\rangle}{\langle\bm{\alpha}|\psi\rangle^*}.
    \end{align}
    Here, using the completeness of coherent states, Eq.~\eqref{eq:coherent_complete}, we can show that $q(\bm{\alpha})$ is a proper probability distribution (Husimi Q-distribution of the state $|\psi \rangle$~\cite{husimi1940some}).
    Hence, if we sample $\bm{\alpha}$ from $q(\bm{\alpha})$ and average $X(\bm{\alpha})$ over the samples, we obtain an estimate of the amplitude.
    First, sampling from $q(\bm{\alpha})$ can be efficiently performed modewise:
    \begin{align}
        q(\bm{\alpha})=\frac{|\langle \bm{\alpha}|\psi\rangle|^2}{\pi^M}
        =\prod_{i=1}^M\frac{|\langle \alpha_i|\psi_i\rangle|^2}{\pi},
    \end{align}
    which takes $O(M)$ running time.
    Second, the random variable $X(\bm{\alpha})$ can also be efficiently computed as
    \begin{align}
        X(\bm{\alpha})
        =\frac{\langle\phi |\hat{U}|\bm{\alpha}\rangle}{\langle\bm{\alpha}|\psi\rangle^*}
        =\frac{\langle\phi |U\bm{\alpha}\rangle}{\langle\bm{\alpha}|\psi\rangle^*}
        =\prod_{i=1}^M \frac{\langle\phi_i|(U\bm{\alpha})_i\rangle}{\langle\alpha_i|\psi_i\rangle^*},
    \end{align}
    which takes $O(M^2)$ running time due to the matrix-vector multiplication.
    Here, for the second equality, we used Eq.~\eqref{eq:coherent}.
    Finally, we show that the variance of $X(\bm{\alpha})$ is bounded as
    \begin{align}
        \text{Var}(X)
        \leq \mathbb{E}_{\bm{\alpha}}\left[\left|\frac{\langle\phi |\hat{U}|\bm{\alpha}\rangle}{\langle\bm{\alpha}|\psi\rangle^*}\right|^2\right]
        =\frac{1}{\pi^M}\int d^{2M}\alpha |\langle\phi |\hat{U}|\bm{\alpha}\rangle|^2
        =1.
    \end{align}
    Similarly to the proof of Theorem 1, the median-of-means estimator gives an $\epsilon$ error with probability $1-\delta$ with samples $O(M^2\log\delta^{-1}/\epsilon^2)$.
    Hence, the entire procedure takes running time $O(M^2\log\delta^{-1}/\epsilon^2)$, which proves the theorem.
\end{proof}

\section{Proof of Corollary 2}
In this section, we provide the proof of Corollary 2 in the main text, which is restated here:
\renewcommand{\thecorollary}{2}
\begin{corollary}[Hafnian approximation]\label{co:haf}
    The hafnian of an $M \times M$ complex symmetric matrix $R$ can be approximated within additive-error $\epsilon \norm{R}^{M/2}$ with probability $1-\delta$ in $O(M^3\log \delta^{-1}/\epsilon^2)$ time.
\end{corollary}

\begin{proof}
We start with the general relation between an output probability amplitude of an $M$-mode Gaussian boson sampling circuit and the hafnian of an $M \times M$ matrix ($M$ is even)~\cite{barvinok2016combinatorics, hamilton2017gaussian,quesada2019franck}.
If we prepare $M$ squeezed vacuum states with squeezing parameter $r_i$'s for each $i$th modes, i.e., $|\bm{r}\rangle =\otimes_{i=1}^M |r_i\rangle$ with squeezed vacuum states $| r_i \rangle$, and apply a linear-optical circuit $\hat{U}$, characterized by an $M\times M$ unitary matrix $U$ (see Eq.~\eqref{eq:LON}), and consider an output probability amplitude on the projector of the single-photon Fock state $|\bm{1}\rangle=\otimes_{i=1}^M |1 \rangle$ with the single-photon Fock state $|1 \rangle$, the amplitude $\langle \bm{1}|\hat{U}|\bm{r}\rangle$ is expressed by the hafnian of an $M\times M$ complex symmetric matrix $R'$:
\begin{align}\label{eq:hafamp}
    \text{haf}(R')=\frac{1}{Z^{1/2}}\langle \bm{1}|\hat{U}|\bm{r}\rangle,
\end{align}
where the matrix $R'$ is related to the parameters of the linear-optical circuits as $R'=UDU^T$ with $D=\oplus_{i=1}^M\tanh{r_i}$ and $Z=\prod_{i=1}^M\cosh{r_i}$. 
Note that, conversely, any complex symmetric matrix can be expressed by Takagi's decomposition~\cite{takagi1924algebraic}, the same form of $UDU^T$ with $D\geq0$. 
This means that we can encode any complex symmetric matrix $R'$ with the largest singular value $\lambda'_\text{max} <1$ (because $\tanh r_i<1$) into an amplitude of a Gaussian boson sampling circuit using Takagi's decomposition, which takes $O(M^3)$ time.

Now suppose an arbitrary complex symmetric matrix $R$ has singular values greater than 1. In this case, we rescale it by $R'=R/(a\lambda_\text{max})$ with the largest singular value $\lambda_\text{max}$ and a constant $a>1$. Then $\text{haf}(R)=(a\lambda_\text{max})^{M/2}\text{haf}(R')$. From the estimation of the amplitude $\langle \bm{1}|\hat{U}|\bm{r}\rangle$, i.e., Theorem~2 with $|\psi\rangle=|\bm{r}\rangle$, $|\phi\rangle=|\bm{1}\rangle$,
\begin{align}
    \text{Pr}\left[ |\mu'-\text{haf}(R)|\geq (a\lambda_\text{max})^\frac{M}{2}Z^{-\frac{1}{2}}\epsilon \right]\leq \delta,
\end{align}
where $\mu'=(a\lambda_\text{max})^\frac{M}{2}Z^{-\frac{1}{2}}\mu$ with the median-of-means estimator $\mu$ of the amplitude $\langle \bm{1}|\hat{U}|\bm{r}\rangle$. Since 
\begin{align*}
    Z^{-\frac{1}{2}}=\prod_{i=1}^M\frac{1}{\sqrt{\cosh r_i}}\leq 1,
\end{align*}
the estimator $\mu'$ of $\text{haf}(R)$ has additive-error at most $\epsilon \lambda_\text{max}^{M/2}=\epsilon\norm{R}^{M/2}$ with probability $1-\delta$ by taking $a \rightarrow 1$ and using $O(\log \delta^{-1}/\epsilon^2)$ samples; thus the time complexity for approximating $\text{haf}(R)$ is $O(M^3\log \delta^{-1}/\epsilon^2)$.

\end{proof}

\section{Generalized molecular vibronic spectra problem}
In this section, we provide more details about a generalized version of the molecular vibronic spectra problem~\cite{huh2015boson}, which has been considered a feasible application of boson sampling.
This problem is essentially equivalent to computing grouped output probabilities of boson sampling, where the way of grouping follows a linear structure:
\begin{align}\label{eq:group}
    G(\Omega)\equiv \sum_{\bm{m}=\bm{0}}^{\infty}p(\bm{m})\delta(\Omega-\bm{\omega}\cdot \bm{m}),
\end{align}
where $\bm{\omega} \in \mathbb{Z}^M_{\geq0}$, $\Omega=\{0,\dots,\Omega_\text{max}\}$ with a cutoff $\Omega_\text{max}$, and $p(\bm{m})=|\langle \bm{m}|\hat{U}|\psi\rangle|^2$ is the probability of obtaining photon number outcome $\bm{m}$ from the given boson sampling circuit $\hat{U}$ with input state $|\psi\rangle$.
Thus, $G(\Omega)$ is the sum of the output probabilities $p(\bm{m})$ that satisfy the condition that $\bm{\omega}\cdot\bm{m}=\Omega$.
This problem has been considered a promising problem that boson sampling can solve more efficiently than classical means because a naive way of computing the grouped output probabilities requires us to compute the output probabilities $p(\bm{m})$, which is \#P-hard problem in general.
Furthermore, because we group the outcomes, the number of outcomes corresponding to each $G(\Omega)$ can be exponentially many, which also requires an exponential cost.
On the other hand, boson sampling devices can easily approximate the grouped output probability by simply running the boson sampling and collecting outcomes following the output probability $p(\bm{m})$ and then computing the frequency of the outcomes corresponding to each $\Omega$.
Therefore, boson sampling solves a problem that seems to take exponential time for classical computers.

However, very recently, a classical algorithm has been developed in Ref.~\cite{oh2024mvs} that can solve the same problem in the same performance.
The main observation is that because boson sampling is a sampling device, the accuracy of its estimation of the grouped output probability is inevitably limited by the finite sample size.
Also, Ref.~\cite{oh2024mvs} found that the Fourier transform of Eq.~\eqref{eq:group} is written as
\begin{equation}\label{eq:fourier}
    \tilde{G}(k)
    =\langle \psi|\hat{U}^{\dagger}e^{ik\theta \hat{\bm{n}}\cdot \bm{\omega}}\hat{U}|\psi\rangle,
\end{equation}
where $\theta\equiv 2\pi/(\Omega_\text{max}+1), k \in \{ 0,\dots,\Omega_\text{max}\}$ and $\hat{\bm{n}}\equiv(\hat{n}_1,\dots,\hat{n}_M)$ is the $M$-mode number operator. By applying the inverse Fourier transform, one can efficiently estimate the grouped probabilities $G(\Omega)$ if Eq.~\eqref{eq:fourier} can be efficiently approximated.

In fact, Ref.~\cite{oh2024mvs} has shown that Eq.~\eqref{eq:fourier} can be exactly computed for Gaussian boson sampling and approximated within additive-error $\epsilon$ in polynomial time for Fock boson sampling; it indicates that there are no quantum advantages from these problems.
On the other hand, the method in Ref.~\cite{oh2024mvs} is not enough to cover more general input states, such as squeezed Fock states.
Hence, whether such general states provide quantum advantages over classical algorithms has remained open.
In the main text, we solve this problem by showing that Eq.~\eqref{eq:fourier} can be efficiently approximated for arbitrary input states, in Corollary 3.

\renewcommand{\thetheorem}{3}
\section{Proof of theorem 3}
In this section, we provide the proof of Theorem 3 in the main text, which is restated here:
\begin{theorem}[Expectation value approximation in near-Clifford circuit]\label{th:nc}
    Consider an $M$-qubit near-Clifford circuit $\hat{U}_{\text{NC}}$ with depth $O(\text{poly}(M))$ and an operator $\hat{O}$.
    The expectation values $\langle \psi | \hat{U}^\dagger_{\text{NC}}\hat{O}\hat{U}_{\text{NC}} |\psi\rangle$ can be approximated within additive-error $\epsilon$ with probability $1-\delta$ in running time $O(\text{poly}(M)\lVert\hat{O}_{A}\rVert^2_2\Tr[\hat{\rho}_{A}]^2 \log(1/\delta)/\epsilon^2)$.
\end{theorem}
\begin{proof}
First, we consider a Clifford circuit $\hat{U}_C$. Using the Pauli operator expansion, $\hat{O}_A=2^{-l}\sum_{\bm{a}\in \{0,1,2,3\}^{l}}\chi_{\hat{O}_A}(\bm{a})\hat{P}_{\bm{a}}$, where $\chi_{\hat{Q}}(\bm{a})\equiv \Tr[\hat{Q}\hat{P}_{\bm{a}}]$ and $\hat{P}_{\bm{a}}\in \{I,X,Y,Z\}^l$ is a Pauli operator, the expectation value of $\hat{O}=\hat{O}_{A}\otimes \hat{\mathbb{1}}_{B}$ can be written as 
\begin{align}
    \langle \psi | \hat{U}^\dagger_\text{C} (\hat{O}_A\otimes \hat{\mathbb{1}}_B)\hat{U}_\text{C} | \psi \rangle
    &=\frac{1}{2^l}\sum_{\bm{a} \in \{ 0,1,2,3\}^l}\chi_{\hat{O}_A}(\bm{a})\chi_{\hat{\rho}_A}(\bm{a}) \\
    &=\sum_{\bm{a} \in \{ 0,1,2,3\}^l}\frac{|\chi_{\hat{O}_A}(\bm{a})|^2}{2^l\|\hat{O}_A\|_2^2}\frac{\|\hat{O}_A\|_2^2\chi_{\hat{\rho}_A}(\bm{a})}{\chi^*_{\hat{O}_A}(\bm{a})} \\
    &\equiv \sum_{\bm{a} \in \{ 0,1,2,3\}^l}q(\bm{a})X(\bm{a}),
\end{align}
where we defined $\hat{\rho}_A\equiv \Tr_B[\hat{U}_\text{C}|\psi\rangle\langle \psi|\hat{U}_\text{C}^\dagger]$,
\begin{align}
    q(\bm{a})=\frac{|\chi_{\hat{O}_{A}}(\bm{a})|^2}{2^l\|\hat{O}_{A}\|^2_2}~~~\text{and}~~~X(\bm{a})=\frac{\|\hat{O}_{A}\|^2_2\chi_{\hat{\rho}_A}(\bm{a})}{\chi^*_{\hat{O}_{A}}(\bm{a})}.
\end{align}

Using the Pauli twirling, one can show that $q(\bm{a})$ is a proper probability distribution and we can efficiently sample $\bm{a}$ by qubit-wise sampling:
\begin{align}
    q(\bm{a})
    =\frac{|\chi_{\hat{O}_A}(\bm{a})|^2}{2^l\|\hat{O}_{A}\|^2_2}
    =\prod_{i\in A}\frac{|\chi_{\hat{O}_{i}}(a_i)|^2}{2\|\hat{O}_{i}\|^2_2},
\end{align}
which takes $O(M)$ running time.

Also, the random variable $X(\bm{a})$ can be efficiently computed since Pauli operators are transformed to other Pauli operators by Clifford circuits as $\hat{U}_\text{C}^\dagger P_{\bm{a}}\hat{U}_\text{C}=P_{\bm{a}'}$, which takes $O(\text{poly}(M))$ for a polynomial-depth circuit.
More specifically,
\begin{align}
    X(\bm{a})
    =\frac{\|\hat{O}_{A}\|^2_2\chi_{\hat{\rho}_A}(\bm{a})}{\chi^*_{\hat{O}_{A}}(\bm{a})}
    =\prod_{i\in A}\frac{\|\hat{O}_{i}\|^2_2}{\chi^*_{\hat{O}_{i}}(a_i)}\langle\psi|\hat{U}^\dagger_\text{C}(\hat{P}_{\bm{a}}\otimes \hat{\mathbb{1}}_B)\hat{U}_\text{C}|\psi\rangle
    =\prod_{i\in A}\frac{\|\hat{O}_{i}\|^2_2}{\chi^*_{\hat{O}_{i}}(a_i)}\langle\psi|\hat{P}_{\bm{a}'}|\psi\rangle
    =\prod_{i\in A}\frac{\|\hat{O}_{i}\|^2_2}{\chi^*_{\hat{O}_{i}}(a_i)} \prod_{i=1}^M\langle\psi_i|\hat{P}_{a'_i}|\psi_i\rangle.
\end{align}

Again, we employ the median-of-means estimator, and since the variance of $X$ is bounded as
\begin{align}
    \text{Var}(X)\leq \mathbb{E}_{\bm{a}}\left[\left|\frac{\|\hat{O}_{A}\|^2_2\chi_{\hat{\rho}_A}(\bm{a})}{\chi^*_{\hat{O}_{A}}(\bm{a})}\right|^2\right]
    =\frac{\|\hat{O}_{A}\|^2_2}{2^l}\sum_{\bm{a}}|\chi_{\hat{\rho}_A}(\bm{a})|^2
    =\|\hat{O}_{A}\|^2_2\|\hat{\rho}_{A}\|^2_2,
\end{align}
if we use $O(\|\hat{O}_A\|_2^2\|\hat{\rho}_A\|_2^2\log\delta^{-1}/\epsilon^2)$ samples, the median-of-means estimator $\mu$ satisfies the following:
\begin{align}
    \Pr[|\mu-\langle \psi|\hat{U}_\text{C}^\dagger (\hat{O}_{A}\otimes \mathbb{1}_{B})\hat{U}_\text{C}|\psi\rangle |\geq \epsilon]\leq \delta,
\end{align}
which gives the theorem for Clifford circuits.

Furthermore, our algorithm is still valid for near-Clifford circuits, i.e., $\hat{U}_\text{C}$$+$logarithmic number of $T$-gates. 
This is because a single Pauli operator is transformed into the sum of at most two Pauli operators after passing through each $T$-gate, that is, $\hat{T}^\dagger \hat{P}_0\hat{T}=\hat{P}_0,~\hat{T}^\dagger \hat{P}_1\hat{T}=(\hat{P}_1+\hat{P}_2)/\sqrt{2},~\hat{T}^\dagger \hat{P}_2\hat{T}=(\hat{P}_2-\hat{P}_1)/\sqrt{2},~\hat{T}^\dagger \hat{P}_3\hat{T}=\hat{P}_3$. If there are $t$ of $T$-gates in a near-Clifford circuit $\hat{U}_\text{NC}$, we have at most $M2^t$ Pauli operators to track the transformations.
Therefore, if $t=O(\log M)$, we can handle all the Pauli operators in $O(\text{poly}(M))$ time.
\end{proof}

\section{Median-of-means estimator for complex variable}\label{sec:MoM}
For completeness, we extend the median-of-means estimator of real variables to that of complex variables.
Our median-of-means estimator for complex variables is constructed by obtaining the median-of-mean estimator for the real part and imaginary part separately, $\mu_X$ and $\mu_Y$, and then combining them as $\mu=\mu_X+i\mu_Y$.
For real random variable $X$, let us recall that for $N\geq 68\log(\frac{2}{\delta})\frac{\text{Var}(X)}{\epsilon^2}$ number of samples, the median-of-means estimator $\mu_X$ satisfies~\cite{jerrum1986random}:
\begin{align}
    \Pr[|\mu_X-\mathbb{E}[X]|\geq \epsilon]\leq \delta.
\end{align}

Now, suppose we want to estimate the expectation value of a complex variable $Z$ whose variance is given as
\begin{align}
    \text{Var}(Z)=\mathbb{E}[|Z-\mu|^2].
\end{align}
We note that the complex variable $Z$ can be decomposed as $Z=X+iY$, where $X$ and $Y$ are real random variables that may be correlated with each other.
Here,
\begin{align}
    \text{Var}(Z)
    =\mathbb{E}[|Z-\mathbb{E}[Z]|^2]
    =\mathbb{E}[|X-\mathbb{E}[X]|^2+|Y-\mathbb{E}[Y]|^2]
    =\text{Var}(X)+\text{Var}(Y).
\end{align}
Let us consider the median-of-means estimator $\mu_X$ and $\mu_Y$ of $\mathbb{E}[X]$ and $\mathbb{E}[Y]$, respectively, using $N\geq 68\log(\frac{2}{\delta})\frac{\text{Var}(X)}{\epsilon^2}$ for each.
Then, with the total number of samples is $N\geq 136\log(\frac{2}{\delta})\frac{\text{Var}(X)}{\epsilon^2}$, we have
\begin{align}
    \Pr[|\mu_X-\mathbb{E}[X]|\geq \epsilon]\leq \delta,~~~ \Pr[|\mu_Y-\mathbb{E}[Y]|\geq \epsilon]\leq \delta.
\end{align}
Now, using the union bound and triangle inequality,
\begin{align}
    \Pr[|\mu-\mathbb{E}[Z]|< 2\epsilon]
    &= \Pr[|(\mu_X+i\mu_Y)-(\mathbb{E}[X]+i\mathbb{E}[Y])|< 2\epsilon] \\
    &\geq \Pr[|\mu_X-\mathbb{E}[X]|+|\mu_Y-\mathbb{E}[Y]||< 2\epsilon] \\
    &\geq \Pr[|\mu_X-\mathbb{E}[X]|< \epsilon~\text{and}~|\mu_Y-\mathbb{E}[Y]|< \epsilon].
\end{align}
Thus,
\begin{align}
    &\Pr[|\mu-\mathbb{E}[Z]|\geq 2\epsilon] \\
    &=1-\Pr[|\mu-\mathbb{E}[Z]|< 2\epsilon] \\
    &\leq 1-\Pr[|\mu_X-\mathbb{E}[X]|< \epsilon~\text{and}~|\mu_Y-\mathbb{E}[Y]|< \epsilon] \\
    &\leq \Pr[|\mu_X-\mathbb{E}[X]|\geq \epsilon]+\Pr[|\mu_Y-\mathbb{E}[Y]|\geq \epsilon]
    \leq 2\delta.
\end{align}
Hence, 
\begin{align}
    \Pr[|\mu-\mathbb{E}[Z]|\geq 2\epsilon]\leq 2\delta.
\end{align}
Therefore, by using
\begin{align}
    N\geq 544\log(\frac{1}{\delta})\frac{\text{Var}(Z)}{\epsilon^2},
\end{align}
our estimator of $\mathbb{E}[Z]$ satisfies
\begin{align}
    \Pr[|\mu-\mathbb{E}[Z]|\geq \epsilon]\leq \delta.
\end{align}

\section{Classical simulation of sparse output probability distributions}
In this section, we state classical simulability theorems for sparse output distributions, following Refs.~\cite{kushilevitz1991learning, schwarz2013simulating}.
Before stating the theorems, we first introduce some terminology.
Let $\mathcal{X}\equiv \{0,\dots,d-1\}^M$.
\begin{itemize}
    \item We say that $\mathcal{P}$ is {\it $\epsilon$-approximately $t$-sparse} if there exists a $t$-sparse vector $v$ such that $\sum_{x\in \mathcal{X}}|p_x-v_x|\leq \epsilon$ (i.e. in terms of total variation distance).
    \item A function $f:\{0,1,\dots, d-1\}^M\to \mathbb{C}$ is said to be {\it additively approximable} if there exists a randomized classical algorithm with runtime $\poly(M,1/\epsilon,\log\frac{1}{\delta})$ which, on input from $\mathcal{X}$, outputs with probability at least $1-\delta$ an $\epsilon$-additive approximation of $f(x)$.
\end{itemize}

The relevance of $t$-sparsity is captured by the following lemma.
\renewcommand{\thelemma}{S1}
\begin{lemma} (\cite{kushilevitz1991learning, schwarz2013simulating})
    Let $\mathcal{P}=\{p_x:x\in \mathcal{X}\}$ be an $\epsilon$-approximately $t$-sparse probability distribution.
    Define $\mathcal{X}_{\epsilon_,t}$ to be the subset of $\mathcal{X}$ such that $p_x\geq \epsilon/t$.
    Define $\mathcal{Q}_{\epsilon,t}$ to be the restriction of $\mathcal{P}$ to $\mathcal{X}_{\epsilon,t}$.
    Then $\|\mathcal{Q}_{\epsilon,t}-\mathcal{P}\|\leq 2\epsilon$.
\end{lemma}

Let $\mathcal{P}=\{p_x:x\in \mathcal{X}\}$ be a probability distribution, and let $\mathcal{P}_m$ denote the marginal probability distribution of the first $m$ modes (or qubits).
Suppose all distributions $\mathcal{P}_m$ are additively approximable, i.e., all the marginal probabilities admit efficient additive approximation.
Then there is a randomized classical algorithm which outputs a list of large probabilities, which is highly based on the Kushilevitz-Mansour algorithm introduced in Ref.~\cite{kushilevitz1991learning}:
\renewcommand{\thetheorem}{S1}
\begin{theorem}(Finding large output probabilities~\cite{kushilevitz1991learning, schwarz2013simulating})\label{thm:peak}
    Suppose all distributions $\mathcal{P}_m$ are additively approximable.
    Then, for given $\theta, \pi>0$, there is a randomized classical algorithm which outputs a list of large probabilities, i.e., a list $L = {x_1, . . . , x_l}$ where $l \leq  2/\theta$ and where each $x_i$ is length-$M$ string over $\{0,\dots,d-1\}$ such that, with probability at least $1-\pi$: (a) for all $y \in L$, it holds that $p(y) \geq \theta/2$;
    (b) every length-$M$ string $x$ over $\{0,\dots,d-1\}$ satisfying $p(x) \geq \theta$ belongs to the list $L$.
\end{theorem}
Hence, condition (a) guarantees that all the outputs of this algorithm are heavy outcomes, i.e., outcomes that have a large probability, and condition (b) guarantees that the algorithm does not miss any heavy outcomes.
Using this algorithm, we can run an efficient classical algorithm that approximately \emph{samples} from the output distribution:
\renewcommand{\thetheorem}{S2}
\begin{theorem}(Classical simulation of sparse output distribution~\cite{kushilevitz1991learning, schwarz2013simulating})\label{thm:sparse}
    Let $\mathcal{P}$ be a distribution on $\mathcal{X}$ which satisfies the following conditions:
    (i) $\mathcal{P}$ is $\epsilon$-approximately $t$-sparse, where $\epsilon<1/6$.
    (ii) $\mathcal{P}$ and its marginals are additively approximable.
    Then there exists a randomized classical algorithm with runtime $\poly(M,t,1/\epsilon,\log 1/\delta)$ which outputs (by means of listing all nonzero probabilities) an $s$-sparse probability distribution $\mathcal{P}'=\{p_x'\}$ where $s=O(t/\epsilon)$ such that, with probability at least $1-\delta$, $\mathcal{P}'$ is $O(\epsilon)$-close to $\mathcal{P}$ (more precisely, $12\epsilon$-close). Furthermore, $p_x'\geq \epsilon/8t$ for all $p_x'$ which are nonzero.
    Finally, it is possible to sample from $\mathcal{P}'$ on a classical computer in $\poly(M,t,1/\epsilon,\log 1/\delta)$.
\end{theorem}
Therefore, such results show that if (i) a given output probability distribution is promised to be \emph{sparse} and (ii) its all marginal probabilities are \emph{additively approximable}, then there is an efficient classical algorithm that can efficiently samples from the distribution.

A natural question is how to detect whether a given output probability distribution is sparse in the first place, i.e., how to diagnose the sparsity of the given probability distribution.
To this end, one can apply Theorem~\ref{thm:peak} to recover all outcomes whose probabilities exceed a desired threshold.
Since Theorem~\ref{thm:peak} does not assume sparsity of distribution $\mathcal{P}$, it provides a practical diagnostic: by enumerating all heavy outcomes, one can assess sparsity (e.g., via the size of $L$ relative to $\theta$).
Therefore, the condition (ii) is sufficient for this purpose.

We now discuss the implications of our classical algorithm for linear-optics in light of these results.
First, Theorem~1 in the main text establishes that output probabilities of linear-optical circuits with arbitrary product inputs are {\it additively approximable}.
Consequently, if the output probability distribution of a linear-optical circuit is $\epsilon$-approximately $t$-sparse with $t=O(\poly(M))$, then combining the Kushilevitz-Mansour algorithm with our algorithm yields an $O(\epsilon)$-close classical simulation of the distribution in $\poly(M,1/\epsilon,\log(1/\delta))$.
This proves Corollary~1 in the main text.
Also thanks to Theorem~\ref{thm:peak}, we can easily diagnose the sparsity of the output probability distribution of a given linear-optical circuit.

\section{Comparison of Theorem~3 with previous results}
In this section, we compare our result (Theorem 3) with previous results in the literature, including those for linear-optical circuits.
Throughout this comparison, simulation refers to the task of sampling from the output probability distribution.
We assume that the input is an arbitrary product state and that measurements are performed in the Fock basis for linear-optical circuits, or in the computational basis for Clifford and IQP circuits.
A summary of the comparison is provided in Table 1.

\begin{enumerate}
\item
Exact and approximate simulation

Linear-optical, (near-)Clifford, and IQP circuits cannot be efficiently simulated on a classical computer under the standard assumption that the polynomial hierarchy (PH) does not collapse~\cite{aaronson2011linear,jozsa2013classical,bremner2011classical}.
This hardness remains robust even when allowing additive-errors of  
$O$(1/poly) in total variation distance between the true and simulated probability distributions~\cite{aaronson2011linear,pashayan2020estimation,bremner2016average}.
Since Clifford circuits with stabilizer inputs can be exactly simulated~\cite{aaronson2004improved}, the presence of arbitrary (non-stabilizer) product inputs providing a “magic” resource is a key requirement for classical intractability.

\item
Estimation of output probabilities with multiplicative errors

For linear-optical circuits with photon-number measurements, each output probability is proportional to a matrix function such as the permanent or hafnian, the multiplicative approximation of which is \#P-hard~\cite{aaronson2011linear}.
For Clifford and IQP circuits with arbitrary product inputs, the output probabilities can be expressed in terms of the gap of a Boolean function, whose multiplicative approximation is also \#P-hard~\cite{jozsa2013classical,bremner2011classical}.

\item 
Estimation of output probabilities with additive errors

When the input consists of Fock states, the output probability of a linear-optical circuit is given by the permanent, whose additive-error approximation can be computed efficiently by Gurvits’s algorithm~\cite{gurvits2005complexity}.
Our algorithm (Theorem 1) generalizes this result to arbitrary product input states.
For Clifford and IQP circuits, efficient additive-error approximations of output probabilities are also known~\cite{pashayan2020estimation,shepherd2010binary}.

\item Expectation values

In this work, we provide efficient classical algorithms for estimating expectation values of observables with polynomially bounded Hilbert-Schmidt norms in both linear-optical and (near-)Clifford circuits.
In particular, when the observable is a rank-1 projector, our result implies that the output probability for an arbitrary product input and rank-1 projective measurement can be efficiently approximated within an additive error.
For Clifford circuits, this further implies that two layers of magic resources, where one is at the input and one is immediately before the measurement, can be incorporated while still allowing efficient classical approximation.

\end{enumerate}
\begin{table}
       \centering
       \begin{tabular}{c|c|c|c}
            & linear-optical  & (near-)Clifford& IQP \\ \hline
          Exact simulation  & PH~\cite{aaronson2011computational}&PH~\cite{jozsa2013classical} &   PH~\cite{bremner2011classical} \\ \hline
          Approximate simulation  & PH~\cite{aaronson2011computational} & PH~\cite{pashayan2020estimation} & PH~\cite{bremner2016average} \\ \hline
          Output probability (multiplicative)  & \#P-hard~\cite{aaronson2011computational} & \#P-hard~\cite{jozsa2013classical} &  \#P-hard~\cite{bremner2011classical} \\ \hline 
          Output probability (1/poly-additive) &efficient* &efficient~\cite{bravyi2019simulation}  & efficient~\cite{shepherd2010binary}\\ \hline
          Expectation value (up to poly-$\Vert \hat{O} \Vert_2$)& efficient* & efficient* & - \\ \hline
       \end{tabular}
       \caption{Assuming that the input is an arbitrary product state and measurement is performed on Fock state basis (linear-optical circuits) or computational basis (Clifford/IQP circuits). PH means hard under assumption of non-collapse of polynomial hierarchy. Asterisk (*) denotes results obtained in this work.}
       \label{tab:placeholder}
   \end{table}

\bibliography{reference}